\documentclass{llncs}

\usepackage{amssymb,amsmath,graphics,graphicx,color}
\newtheorem{observation}{Observation}

\usepackage{graphicx,color}

\usepackage{enumitem}

\newcommand{\dist}{{\rm dist}}

\newcommand{\chord}{{\rm chord}}

\begin{document}

\title{Enumeration and Maximum Number of \\ Minimal Connected Vertex Covers in Graphs\thanks{A preliminary version of this paper appeared as an extended abstract in the proceedings of IWOCA 2015.The research leading to these results has received funding from the Research Council of Norway and the European Research Council under the European Union's Seventh Framework Programme (FP/2007-2013) / ERC Grant Agreement n. 267959.}}

\author{
Petr A. Golovach\inst{1}
\and
Pinar Heggernes\inst{1}
\and
Dieter Kratsch\inst{2}
}

\institute{
Department of Informatics, University of Bergen, Norway,
\texttt{\{petr.golovach,pinar.heggernes\}@ii.uib.no}
\and
Universit\'e de Lorraine, LITA, Metz, France, 
\texttt{dieter.kratsch@univ-lorraine.fr}
}

\pagestyle{plain}
\maketitle

\begin{abstract}
{\sc Connected Vertex Cover} is one of the classical problems of computer science, already mentioned in the monograph of Garey and Johnson \cite{GareyJ79}. Although the optimization and decision variants of finding connected vertex covers of minimum size or weight are well studied, surprisingly there is no work on the enumeration or maximum number of minimal connected vertex covers of a graph. In this paper we show that the maximum number of minimal connected vertex covers of a graph is 
at most $1.8668^n$,  
and these can be enumerated in time $O(1.8668^n)$. For graphs of chordality at most 5, we are able to give a better upper bound, and for chordal graphs and distance-hereditary graphs we are able to give tight bounds on the maximum number of minimal connected vertex covers. 
\end{abstract}

\section{Introduction}
\label{sec:intro}
The maximum number of minimal vertex covers that a graph on $n$ vertices can have is equal to the maximum number of maximal independent sets, which is known to be $3^{n/3}$ by a celebrated result of Moon and Moser~\cite{MoonM65}. This result is easily extended to an algorithm that enumerates all the minimal vertex covers of a graph within a polynomial factor of the given bound. The bound is tight as a disjoint union of $n/3$ triangles has exactly $3^{n/3}$ minimal vertex covers. These results have been extremely useful in many algorithms, e.g., they were used by Lawler \cite{Lawler76} to give an algorithm for graph coloring, which was
the fastest algorithm for this purpose for decades. For special graph classes, better bounds have been obtained, e.g., the tight bound for triangle-free graphs is $2^{n/2}$, given by Hujtera and Tuza \cite{HujteraT93} with combinatorial arguments,  and by  Byskov \cite{Byskov04} algorithmically. Also these results have been useful in several algorithms, e.g., for graph homomorphism \cite{FominHK07}. Although connected vertex covers were defined and studied as early as vertex covers \cite{GareyJ79}, interestingly the maximum number of minimal connected vertex covers in graphs or the enumeration of these have not been given attention.

In this paper, we study exactly these questions, and we give an algorithm for enumerating all minimal connected vertex covers of a graph in time $O(1.8668^n)$. 
We also give an upper bound $1.8668^n$ on the number of such covers a graph can have. 
We provide a lower bound example, which is a graph that has $3^{(n-1)/3}$ minimal connected vertex covers, leaving a gap between these bounds on general graphs. We are able to narrow this gap for graphs of 
chordality at most 5, and almost close the gap for chordal graphs and distance-hereditary graphs. In particular, we show that the maximum number of minimal connected vertex covers in chordal graphs, graphs of chordality at most 5, and distance-hereditary graphs, respectively, is at most $3^{n/3}$,  $1.6181^n$, and $2 \cdot 3^{n/3}$. All our bounds are obtained by enumeration algorithms whose running times correspond to the given bounds up to polynomial factors. 

We would like to emphasize that our motivation for the given bounds and enumeration algorithms is not for fast computation of connected vertex covers of minimum size. In fact, as we will see in the next section, such sets can be computed in time $O(1.7088^n)$ on general graphs. Furthermore, Escoffier et al.~\cite{EGM10} have shown that this problem can be solved in polynomial time on chordal graphs. The problem of computing minimum connected vertex covers is indeed well studied with a large number of published results. These are nicely surveyed in the introduction given by Escoffier et al.~\cite{EGM10}.

Our motivation comes from the background given in the first paragraph, as well as the fact that the study of the maximum number of vertex subsets with given properties is a well established area in combinatorics and graph theory. More recently, exponential time enumeration algorithms for listing such vertex subsets in graphs have become increasingly popular and found many applications \cite{FominK10}. For most of these algorithms, an upper bound on the number of enumerated subsets follows from the running time of the algorithm. Examples of such recent results, both on general graphs and on some graph classes, concern the enumeration and maximum number of minimal dominating sets, minimal feedback vertex sets, minimal subset feedback vertex sets, minimal separators, maximal induced matchings, and potential maximal cliques \cite{BHHSV14,CHHK13,CHHV12,FGPR08,FHKPV14,FV10,GM13,GHKS14}.

\section{Preliminaries}\label{sec:defs}
We consider finite undirected graphs without loops or multiple edges. 
For each of the graph problems considered in this paper, we let $n$ denote the number of vertices and $m$  the number of edges of the input graph.
For a graph $G$ and a subset $U\subseteq V(G)$ of vertices, we write $G[U]$ to denote the subgraph of $G$ induced by $U$. We write $G-U$ to denote $G[V(G)\setminus U]$, and $G-u$ if $U=\{u\}$.
A set $U\subseteq V(G)$ is \emph{connected} if $G[U]$ is a connected graph.
For a vertex $v$, we denote by $N_G(v)$ the \emph{(open) neighborhood} of $v$, i.e., the set of vertices that are adjacent to $v$ in $G$.
The \emph{closed neighborhood} is $N_G[v]=N_G(v)\cup \{v\}$. For a set of vertices $U\subseteq V(G)$, $N_G[U]=\cup_{v\in U}N_G[v]$ and $N_G(U)=N_G[U]\setminus U$.
Two distinct $u,v\in V(G)$ are \emph{false twins} if $N_G(u)=N_G(v)$.
The \emph{distance} $\dist_G(u,v)$ between vertices $u$ and $v$ of $G$ is the number of edges on a shortest $(u,v)$-path.
A path or cycle $P$ is \emph{induced} if it has no \emph{chord}, i.e., there is no edge of $G$ that joins any two vertices of $P$ that are not adjacent in $P$. 
The \emph{chordality},  $\chord(G)$, of a graph $G$ is the length of a longest
induced  cycle in $G$; if $G$ has no cycles, then $\chord(G)=0$.
A set of vertices is an {\it independent set} if there is no edge between any pair of these vertices, and it is a {\it clique} if all possible edges are present between pairs of these vertices. An independent set (clique) is {\it maximal} if no set properly containing it is an independent set (clique).
A set of vertices $S\subset V(G)$ of a connected graph $G$ is a \emph{separator} if $G-S$ is disconnected.
A vertex $v$ is a \emph{cut vertex} of a connected graph $G$ if $\{v\}$ is a separator.
For an edge $uv\in E(G)$, the \emph{contraction} of $uv$ is the operation that replaces $u$ and $v$ by a new vertex adjacent to $(N_G(u)\cup N_G(v))\setminus \{u,v\}$. $G/e$ denotes the graph obtained from $G$ by contracting edge $e$. A graph $G'$ is an \emph{induced minor} of $G$ if $G'$ can be obtained from $G$ by deleting vertices and contracting edges. 

For a non-negative integer $k$, a graph $G$ is \emph{$k$-chordal} if $\chord(G)\leq k$.
A graph is \emph{chordal} if it is $3$-chordal.
A graph is a \emph{split} graph if its vertex set can be partitioned in an independent set and a clique.
A graph is \emph{cobipartite} if its vertex set can be partitioned into two cliques.
A graph $G$ is a \emph{chordal bipartite} graph if $G$ is a bipartite graph and $\chord(G)\leq 4$.
A graph $G$ is \emph{distance-hereditary} if for every connected induced subgraph $H$ of $G$, $\dist_H(u,v)=\dist_G(u,v)$ for $u,v\in V(H)$.  
Each of the above mentioned graph classes can be recognized in polynomial (in most cases linear) time, and they are closed under taking induced subgraphs~\cite{BrandstadtLS99,Golumbic04}. See the monographs by Brandst{\"a}dt et al.~\cite{BrandstadtLS99} and Golumbic \cite{Golumbic04} for more properties and characterizations of these classes and their inclusion relationships.

A set of vertices $U\subseteq V(G)$ is a \emph{vertex cover} of $G$ if for every $uv\in E(G)$, $u\in U$ or $v\in U$. A vertex cover $U$ is \emph{connected} if $U$ is a connected set. A (connected) vertex cover $U$ is \emph{minimal} if no proper subset of $U$ is a (connected) vertex cover. 
Observe that $U$ is a minimal connected vertex cover of $G$ if and only if for
every vertex $u\in U$, either $u$ is a cut vertex of $G[U]$ or there is an edge $ux$ 
of $G$ such that $x\notin U$. 
Hence given a vertex set $U\subseteq V(G)$, it can be decided
in time $O(nm)$ whether $U$ is a minimal connected vertex cover of $G$. 

It is easy to see that $U$ is a (minimal) vertex cover of $G$ if and only if $V(G)\setminus U$ is a (maximal) independent set. The following upper bound for the number of maximal independent sets was obtained by Miller and Muller~\cite{MillerM60} and Moon and Moser~\cite{MoonM65}.

\begin{theorem}[\cite{MillerM60,MoonM65}]\label{thm:vc} 
The number of minimal vertex covers (maximal independent sets) of a graph is at most
$$\begin{cases}
3^{n/3}&\text{if } n\equiv 0~ (\text{\rm mod } 3),\\
4\cdot 3^{(n-4)/3}&\text{if } n\equiv 1~ (\text{\rm mod } 3),\\
2\cdot 3^{(n-2)/3}&\text{if } n\equiv 2~ (\text{\rm mod } 3).
\end{cases}$$
\end{theorem}

Together with the fact that all maximal independent sets can be enumerated with polynomial delay (see, e.g., \cite{TsukiyamaIAS77,JohnsonP88}), this implies that all minimal vertex covers of a graph can be enumerated in time $O^*(3^{n/3})$, where the $O^*$-notation suppresses polynomial factors.
Note that the same result can also be obtained by a branching algorithm (see, e.g.~\cite{FominK10}).

The bounds of Theorem~\ref{thm:vc} are tight; a well known lower bound example is a graph consisting of $n/3$ disjoint triangles, which is a chordal distance-hereditary graph. By adding a vertex which is adjacent to every vertex of this graph, we can obtain a lower bound for the maximum number of minimal connected vertex covers of a graph.

\begin{proposition}\label{prop:lower}
There are chordal distance-hereditary graphs  with at least $3^{(n-1)/3}$ minimal connected vertex covers.
\end{proposition}

\begin{proof}
Consider the graph $G$ constructed as follows, for a positive integer $k$:
\begin{itemize}
\item for $i\in\{1,\ldots,k\}$, construct a clique $T_i=\{x_i,y_i,z_i\}$,
\item construct a vertex $u$ and make it adjacent to $x_i,y_i,z_i$, for $i\in\{1,\ldots,k\}$.
 \end{itemize}
Observe that every minimal connected vertex cover of $G$ contains $u$ and exactly two vertices of each clique $T_i$, for $i\in\{1,\ldots,k\}$, and every set of this type is a minimal connected vertex cover. Therefore, $G$ has $3^{(n-1)/3}$ minimal connected vertex covers. It is easy to check that the given graph is both chordal and distance-hereditary.
\qed
\end{proof}

We do not know of any better lower bounds for the maximum number of minimal connected vertex covers on graphs in general. We will use the following simple observation to give upper bounds on the number of minimal connected vertex covers.

\begin{observation}\label{obs:cuts}
Let $S$ be a separator of a connected graph $G$. Then for every connected vertex cover $U$ of $G$, $S\cap U\neq \emptyset$. In particular, if $v$ is a cut vertex, then $v$ belongs to every connected vertex cover.
\end{observation}

\begin{proof}
Let $U$ be a connected vertex cover of $G$, and suppose that $S$ is a separator of $G$. 
Then there are vertices $x,y$ such that $x$ and $y$ are in two distinct components of $G-S$, and each of $x$ and $y$ has a  neighbor in $S$.
If $S\cap U=\emptyset$, then $x,y\in U$ and $G[U]$ is disconnected; a contradiction. Hence, $S\cap U\neq\emptyset$. \qed
\end{proof}

Recall that our motivation for enumerating the minimal connected vertex covers of a graph is not for the computation of a connected vertex cover of minimum size. In fact such a set can be computed in time $O(1.7088^n)$, using the following result of Cygan \cite{Cygan12} about the parameterized complexity of the 
problem.

\begin{theorem}[\cite{Cygan12}]\label{thm:cvc-param}
It can be decided in time $O(2^k\cdot n^{O(1)})$ and in polynomial space, whether a graph has a connected vertex cover of size  at most $k$.
\end{theorem}

Combining the algorithm of Cygan~\cite{Cygan12} with brute force checking of all vertex subsets of size at most $k$, we obtain the following corollary.

\begin{corollary}\label{cor:cvc-minimum}
It can be decided in time $O(1.7088^n)$ and in polynomial space, whether a graph has a connected vertex cover of size  at most $k$.
\end{corollary}

\begin{proof}
Let $\alpha\in(1/2,1)$ be the unique root of the equation
$$2^{\alpha}=\Big(\frac{1}{\alpha}\Big)^\alpha\cdot\Big(\frac{1}{1-\alpha}\Big)^{1-\alpha}.$$
It can be seen that $\alpha\approx 0.7728$ and $2^\alpha<1.7088$. 
If $k\leq \alpha n$, we use Theorem~\ref{thm:cvc-param} and solve the considered instance by the algorithm of Cygan~\cite{Cygan12} in time $O^*(2^{\alpha n})$. If $k>\alpha n$, then 
we use brute force and for each $U\subseteq  V(G)$ of size $k$, check whether $U$ is a connected vertex cover. 
Because $n/2\leq \alpha n\leq k\leq n$, we have
$\binom{n}{k}\leq 2^{H(1-\alpha)n}$, where
$H(x)=-x \log_2x-(1-x)\log_2(1-x)$
is the \emph{entropy function} (see, e.g.,~\cite{FominK10}). 
By the choice of $\alpha$, $2^{H(\alpha)n}\leq 2^{\alpha n}$, and we solve the problem in time $O^*(2^{\alpha n})$ if $k>\alpha n$.
Since $2^\alpha<1.7088$, the running time of the algorithm is $O(1.7088^n)$.\qed
\end{proof}

The majority of our upper bounds for the number of minimal connected vertex covers will be given via enumeration algorithms that are recursive branching algorithms.  For the analysis of the running time and the number of sets that are produced by such an algorithm, we use a technique based on solving recurrences for branching steps and branching rules respectively. 
We refer to the book \cite{FominK10} for a detailed introduction. 
To analyze such a branching algorithm solving an enumeration problem, one assigns
to each \emph{instance} $I$ of the recursive algorithm a \emph{measure} $\mu(I)$ that one may 
consider as the size of the instance $I$. 
If the algorithm branches on an instance $I$ into $t$ new instances, such that the  
measure decreases by $c_1, c_2, \ldots, c_t$ for each new instance, respectively, we say that
$(c_1, c_2, \ldots, c_t)$ is the \emph{branching vector} of this step, respectively branching rule.
Let us first consider the case  that the branching algorithm has only one branching vector. 
In particular, if $L(s)$ is the maximum number of leaves of a search tree for an instance $I$ of 
measure $s=\mu(I)$, then we obtain the recurrence 
$L(s) \le L(s-c_1) +L(s-c_2) + \ldots + L(s-c_t)$. Then standard analysis (see~\cite{FominK10}), gives us that  $L(s)=O^*(\alpha^s)$,  
where $\alpha$ is the unique positive real root of $x^c-x^{c-c_1}- \ldots - x^{c-c_t} = 0$ for $c=\max\{c_1,\ldots,c_t\}$. 
If $\mu(I)\leq n$ for all instances $I$, then the number of leaves of the search tree is $O^*(\alpha^n)$.
The number $\alpha$ is called the {\it branching number} of this branching vector. 
Now let us consider the general case  in which different branching vectors are involved at different steps of an algorithm.  Then the  branching vector with the highest branching number gives an upper bound on $L(s)$. 
This approach allows us to  achieve running times of the form $O^*(\alpha^n)$ for some real $\alpha \ge 1$. As the number of minimal connected vertex covers produced by an algorithm is upper bounded by the number of leaves of the search tree, we also obtain the upper bound for the number of minimal connected vertex covers of the same form $O^*(\alpha^n)$. If $\alpha$ has been obtained by rounding then one may replace $O^*(\alpha^n)$ by $O(\alpha^n)$.

Now we intend to provide a tool that allows to directly prove the following stronger 
statement: for all $n\geq 1$ the number of minimal connected vertex covers in an $n$-vertex graph of a certain graph class is at most $\alpha^n$. Actually this approach is much more general and can be used for many enumeration problems solved by branching algorithms. 
We start with a lemma providing an upper bound for the solution of the recurrences in
recursive branching algorithms. 

\begin{lemma}\label{lem:recurrences}
Let $\{(c_1^{(i)},\ldots, c_{t_i}^{(i)})\mid i\in J\}$ be a possibly infinite collection of vectors of positive integers.
Suppose that  $L\colon\mathbb{Z}\rightarrow\mathbb{N}_0$ is a function such that 
\begin{itemize}
\item[i)] $L(k)=0$ if $k<0$, and
\item[ii)] $L(k)\leq \max\{1, \max\{L(k-c_1^{(i)})+\ldots+L(k-c_{t_i}^{(i)})\mid i\in J \}\}$ 
for $k\geq 0$.
\end{itemize}
Then for $k\geq 0$, 
$$L(k)\leq \alpha^k,$$
where $\alpha=\max\{\alpha_i \mid i\in J\}$ and 
for $i\in J$, $\alpha_i\geq 1$ is the unique positive root of 
$$x^{c^{(i)}}-x^{c^{(i)}-c_1^{(i)}}- \ldots - x^{c^{(i)}-c^{(i)}_{t_i}} = 0,$$ 
where $c^{(i)}=\max\{c_1^{(i)},\ldots,c_{t_i}^{(i)}\}$. 
\end{lemma}

\begin{proof}
The proof is done by induction on $k$. 

By $ii)$, we have that $L(k)\leq 1 \le \alpha^k$ holds for $k=0$. Notice also that $L(k)=0\leq \alpha^k$ for $k<0$ by i). Let $k\geq 1$ and assume inductively that the claim holds for the smaller values of $k$. By ii), $L(k)\leq L(k-c_1^{(i)})+\ldots+L(k-c_{t_i}^{(i)})$ for some $i\in J$.
If $L(k)\leq 1$, then $L(k)\leq\alpha^k$. Assume that $L(k)>1$.
By ii), $L(k)\leq L(k-c_1^{(i)})+\ldots+L(k-c_{t_i}^{(i)})$ for some $i\in J$.
By induction, $$L(k)\leq \alpha^{k-c_1^{(1)}}+\ldots+\alpha^{k-c_{t_i}^{(i)}}.$$
Recall that $\alpha\geq\alpha_i$ and, by the definition of $\alpha_i$, 
$$\alpha^{c^{(i)}}\geq \alpha^{c^{(i)}-c_1^{(i)}}+ \ldots + \alpha^{c^{(i)}-c^{(i)}_{t_i}},$$ 
and, therefore, 
$$\alpha^k\geq \alpha^{k-c_1^{(i)}}+\ldots +\alpha^{k-c^{(i)}_{t_i}}.$$ 
Hence, $L(k)\leq \alpha^k$.\qed
\end{proof}

Let us summarize the typical setting in which Lemma~\ref{lem:recurrences} will be applied
to upper bound the number of certain objects, for example minimal connected vertex
covers, in an $n$-vertex graph. 

\begin{lemma}\label{lem:upper-bound}
Let $\mathcal{A}$ be a branching algorithm enumerating all objects of property $P$ of an
input graph $G$. Let $L(s)$ be the maximum number of leaves of a search tree rooted
at an instance of measure $s$.
Suppose the following conditions are satisfied.
\begin{itemize}
\item[a)]  There is a measure $\mu$ assigning to each instance of the algorithm an integer
such that $\mu(I) \le n$ for all instances. 
\item[b)]  $L(0)=1$ for $s=0$.
\item[c)]  $L(s)=0$ for all $s<0$.
\item[d)]  The recurrences  corresponding to the branching vectors of  $\mathcal{A}$ 
can be written in the form of condition ii) of the Lemma~\ref{lem:recurrences}.
\end{itemize}
Then  
the number of leaves of the search tree and thus the number of objects of property $P$ 
 is at most $\alpha^n$, where $\alpha$ is the 
largest branching number. 
\end{lemma}

Lemma~\ref{lem:upper-bound} is an immediate consequence of Lemma~\ref{lem:recurrences}. 
It is worth mentioning  that c) and d) hold if no execution of $\mathcal{A}$ produces instances with negative measures, which is a common feature of branching algorithm. 
Notice that in this case we have that if $\{(c_1^{(i)},\ldots, c_{t_i}^{(i)})\mid i\in J\}$ is the collection of branching vectors for such an algorithm and 
$c=\min_{i\in J}\max\{c_j^{(i)}\mid 1\leq j\leq t_i\} $, then we do not branch on instances $I$ with $\mu(I)<c$. Hence, if $\mu(I)<c$, then the search tree has at most one leaf as required by ii). If $\mu(I)\geq c$, then the number of leaves is upper bounded by the value produced by one of the recurrences if we branch on the instance or by 1 if we do not branch.

\begin{remark}
Lemma~\ref{lem:upper-bound} also holds when the measure $\mu$ assigns to each 
instance of the algorithm a rational such that $\mu(I) \le n$ for all instances. 
See the corresponding statement in~\cite{FGPS09} about the possibility of an 
inductive proof over non-integral measures. 
\end{remark}

\section{General graphs} \label{sec:general}

In this section we give a non-trivial upper bound on the maximum number of minimal 
connected vertex covers in arbitrary $n$-vertex graphs. 

\begin{theorem}\label{thm:gen}
The maximum number of minimal connected vertex covers of an arbitrary graph is 
at most $1.8668^n$,
 and these can be enumerated in time $O(1.8668^n)$.
\end{theorem}

\begin{proof}
We give a branching algorithm that we call $\textsc{EnumCVC}(S,F)$, which takes as the input two disjoint sets $S,F\subseteq V(G)$ and outputs minimal connected vertex covers $U$ of $G$ such that $S\subseteq U\subseteq S\cup F$. We call $\textsc{EnumCVC}(\emptyset,V(G))$ to enumerate the minimal connected vertex covers of $G$.
We say that $v\in V(G)$ is \emph{free} if $v\in F$ and $v$ is \emph{selected} if $v\in S$. The algorithm branches on a subset of free vertices and either \emph{selects} some of them to be included in a (potential) minimal connected vertex cover or \emph{discards} some of them  by forbidding them to be selected.

\medskip
\noindent
$\textsc{EnumCVC}(S,F)$
\begin{enumerate}
\item If $S$ is a minimal connected vertex cover then return $S$ and stop. 
\item If $F=\emptyset$, then stop.
\item If there are two adjacent free vertices $u,v\in F$, then branch as follows:
\begin{itemize}
\item select $u$, i.e., set $S'=S\cup\{u\}$, $F'=F\setminus\{u\}$, and call $\textsc{EnumCVC}(S',F')$,
\item discard $u$ and select its neighbors, i.e., set $S'=S\cup N_G(u)$, $F'=F\setminus N_G[u]$, and call  $\textsc{EnumCVC}(S',F')$.
\end{itemize}
\item If $F$ is an independent set, then let $s$ be the number of components of $G[S]$. Consider every non-empty set $X\subseteq F$ of size at most $s-1$ and output $S\cup X$ if $S\cup X$ is a minimal connected vertex cover of $G$.  
\end{enumerate}

To argue that the algorithm is correct, consider a minimal connected vertex cover $U$ of $G$ such that $S\subseteq U\subseteq S\cup F$ and for every $v\in V(G)\setminus(S\cup F)$, $N_G(v)\subseteq S$. 
If $F=\emptyset$, then $U=S$ and the algorithm outputs $U$ on Step~1.
Assume inductively that  $\textsc{EnumCVC}(S',F')$ outputs $U$ for every pair of disjoint sets $S',F'$ such that $S'\subseteq U\subseteq S'\cup F'$ and $|F'|<|F|$.
Clearly, if $S$ is a connected vertex cover of $G$, then $U=S$ by minimality and $U$ is returned by the algorithm on Step~1. If $S$ is not a connected vertex cover, then $F\cap U\neq\emptyset$ and the algorithm does not stop at Step~2. If there are two adjacent free vertices $u,v\in F$, then $u\in U$, or $u\notin U$ and $N_G(u)\subseteq U$, because at least one endpoint of every edge is in $U$. In the first case we have that $S'\subseteq U\subseteq S'\cup F'$, where  $S'=S\cup\{u\}$ and $F'=F\setminus\{u\}$. 
In the second case,  $S'\subseteq U\subseteq S'\cup F'$, where  $S'=S\cup N_G(u)$ and  $F'=F\setminus N_G[u]$. 
By induction, the algorithm outputs $U$ when we call  $\textsc{EnumCVC}(S',F')$. Finally, if $F$ is an independent set and $S$ is not a connected vertex cover of $G$, then 
$U=S\cup X$ for $X\subseteq F$. Because $F$ is independent and $N_G(v)\subseteq S$ for $v\in V(G)\setminus(S\cup F)$, $S$ is a vertex cover of $G$, i.e., the vertices of $X$ are included in $U$ only to ensure the connectivity of $G[U]$. We have that each vertex of $X$ is a cut vertex of $G[U]$. Since $G[S]$ has $s$ components, $|X|\leq s-1$. Therefore, the algorithm outputs $U$. 
To complete the correctness proof, it remains to notice that if $S=\emptyset$ and $F=V(G)$, then $S\subseteq U\subseteq S\cup F$ and $V(G)\setminus(S\cup F)=\emptyset$. Hence, $\textsc{EnumCVC}(\emptyset,V(G))$ outputs $U$. As $U$ is an  arbitrary minimal connected vertex cover, the algorithm outputs all minimal connected vertex covers. 

We have that $\textsc{EnumCVC}(\emptyset,V(G))$ lists all minimal connected vertex covers of $G$. To obtain the upper bound on the number of minimal connected vertex covers of $G$, we upper bound the number of leaves of the search tree produced by the algorithm.

Observe that by executing Steps~1--3, the algorithm either produces a leaf of the search tree or a node of the tree corresponding to a pair of sets $S$ and $F$ such that $F$ is independent. 
We call a node corresponding to such $S$ and $F$ a \emph{sub-leaf} or \emph{$(S,F)$-sub-leaf}. Note that the children of a sub-leaf are leaves of the search tree produced by Step~4.
The only branching rule (Step~3) has branching vector $(1,2)$ since we remove at least one free vertex in the first branch and at least two, i.e., $u$ and $v$, in the second one. This branching vector has branching number $\alpha \approx 1.6181$.
Moreover, it can be shown by the standard analysis (see~\cite{FominK10}) that by executing Steps~1--3 the algorithm produces $O^*(\alpha^h)$ $(S,F)$-sub-leaves such that $h=n-|F|$.  
We use induction to get rid of the factors hidden in the $O^*$-notation.

For an instance $(S,F)$, consider the search tree produced by Steps~1--3. Denote by $N(k)$ the maximum number of descendant nodes with  $n-h$ free vertices of an instance with $|F|=k+n-h$ free vertices. We claim that $N(k)\leq\alpha^k$.  
If $k<0$, then $N(k)=0$. If $k=0$, then the node itself has $n-h$ free vertices and $N(k)=1$. If $k=1$, then only the child corresponding to the first branch of Step~3 has $n-h$ free vertices. Suppose that $k\geq 2$. Then $N(k)\leq N(k-1)+N(k-2)$, because of the branching on Step~3.
By induction, we have that 
$$N(k)\leq N(k-1)+N(k-2)\leq\alpha^{k-1}+\alpha^{k-2}\leq \alpha^k.$$ 
We conclude that $N(h)\leq\alpha^k$, that is by executing Steps~1--3 we produce at most $\alpha^h$ nodes with $n-h$ free vertices and, therefore, at most $\alpha^k$ $(S,F)$-sub-leaves such that $h=n-|F|$.

Now we consider the $(S,F)$-sub-leaves of the search tree and Step~4. 
We have the following two cases.

\noindent
{\bf Case 1.}  $h=n-|F|\geq n/3$. Then $|F|\leq 2n/3$. Clearly, there are at most $2^{n-h}$ sets $X\subseteq F$ of size at most $s-1$ and such an $(S,F)$-sub-leaf  has at most $2^{n-h}$ children. Since there are 
$\alpha^k$
$(S,F)$-sub-leaves with $h=n-|F|$, the total number of children of these nodes is 
$\alpha^h\cdot 2^{n-h}$. 
Since, $h\geq n/3$ and $\alpha<2$,
the number of these children is 
$\alpha^{n/3}\cdot 2^{2n/3}$.

\noindent
{\bf Case 2.} $h=n-|F|<n/3$. Let $s\geq 2$ be the number of components of $G[S]$. We have that $s-1\leq h$. Let $\beta=h/n$.   Then $h=\beta n$, $n-h=(1-\beta)n$ and
$h/(n-h)=\beta/(1-\beta)\leq 1/2$.
The number of non-empty sets $X\subseteq F$ such that $|X|\leq s-1$ is 
$$\binom{n-h}{1}+\ldots+\binom{n-h}{s-1}\leq\binom{(1-\beta)n}{1}+\ldots+\binom{(1-\beta)n}{\beta n} \leq 2^{H(\beta/(1-\beta))(1-\beta)n},$$
where 
$H(x)=-x \log_2x-(1-x)\log_2(1-x)$
is the entropy function (see, e.g.,~\cite{FominK10}). 
Let
$$f(\beta)=\alpha^\beta\cdot 2^{H(\beta/(1-\beta))(1-\beta)}=\Big(\frac{1+\sqrt{5}}{2}\Big)^\beta\cdot\Big(\frac{1-\beta}{\beta}\Big)^\beta\Big(\frac{1-\beta}{1-2\beta}\Big)^{1-2\beta}.$$
The function $f(\beta)$ on the interval $(0,1/3)$ has the maximum 
value\footnote{The computations have been done by computer.} 
for $\beta^*=\frac{1}{2}-\frac{1}{2\sqrt{3+2\sqrt{5}}}$ and $f(\beta^*)\approx 1.86676$. Since the number of $(S,F)$-sub-leaves with $n-|F|=h$ is 
$\alpha^h$
we obtain that the total number of children of these sub-leaves is
$f(\beta^*)^n$.

\medskip
Because  $\alpha\approx 1.61803<1.8668$,  $\alpha^{1/3}\cdot 2^{2/3}\approx 1.86369<1.8668$ and $f(\beta^*)\approx 1.86676<1.8668$, the total number of leaves of the search tree is $1.8668^n$ and, therefore, the number of minimal connected vertex covers is $1.8668^n$.

It remains to observe that each step of $\textsc{EnumCVC}$ can be done in polynomial time. Then the bound on the number of leaves of the search tree  immediately implies that the algorithm runs in time $O(1.8668^n)$.
\qed
\end{proof}

\section{Graphs of chordality at most 5}\label{sec:chord}
The upper bound that we proved in the previous section leaves a gap between that bound and the best known lower bound given in Proposition \ref{prop:lower}. In this section, we will narrow this gap for graphs of chordality at most 5, and we will close the gap for chordal graphs, i.e., graphs of chordality at most 3. We start with this latter class. 

\begin{theorem}\label{thm:chord}
The maximum number of minimal connected vertex covers of a chordal graph is at most $3^{n/3}$, and these can be enumerated in time $O^*(3^{n/3})$.
\end{theorem}

\begin{proof}
Let $G$ be a chordal graph. If $G$ has no edges, then the claim is trivial. Notice also that the removal of an isolated vertex does not influence connected vertex covers, and if $G$ has two components with at least one edge each, $G$ has no connected vertex cover. Hence, without loss of generality we can assume that $G$ is a connected graph and $n\geq 2$.  

Let $S$ be the set of cut vertices of $G$ and $G'=G-S$. We claim that $U\subseteq V(G)$ is a minimal connected vertex cover of $G$ if and only if $S\subseteq U$ and $X=U\cap V(G')$ is a minimal vertex cover of $G'$. 

Let $X$ be a vertex cover of $G'$. We show that $U=S\cup X$ is a connected vertex cover of $G$. Because $X$ is a vertex cover of $G'$ and $S$ covers the edges of $E(G)\setminus E(G')$, $U$ is a vertex cover of $G$. To show that $G[U]$ is connected, assume for the sake of contradiction that it is not so. Let $H_1$ and $H_2$ be distinct components of $G[U]$ at minimum distance from each other. Let $P=v_0\ldots v_k$ be a shortest path in $G$ that joins a vertex of $H_1$ with a vertex of $H_2$. For $i\in\{1,\ldots,k\}$,  $v_{i-1}\in U$ or $v_i\in U$, because $U$ is a vertex cover of $G$. Since $H_1$ and $H_2$ are chosen to be components at minimum distance, $k=2$. Since $v_1\notin U$, $v_1$ is not a cut vertex of $G$. Therefore, $G-v_1$ has  a shortest $(v_0,v_2)$-path $P'=u_0\ldots u_s$. Because $P'$ is an induced path and $v_0v_2\notin E(G)$, $v_1u_i\in E(G)$ for $i\in\{1,\ldots,s-1\}$ by chordality. As $v_1\notin U$, $u_i\in U$ for $i\in\{1,\ldots,u_{s-1}\}$. Therefore, $V(P')\subseteq U$ contradicting that 
$H_1$ and 
$H_2$ are distinct components of $G[U]$. 
Since $U$ is a vertex cover and $G[U]$ is connected, $U$ is a connected vertex cover of $G$. 

Let $U$ be a connected vertex cover of $G$. 
By Observation~\ref{obs:cuts}, $S\subseteq U$. 
As the vertices of $S$ cover only the edges of $E(G)\setminus E(G')$, $X=U\setminus S$ is a vertex cover of $G'$.

We proved that $U\subseteq V(G)$ is a connected vertex cover of $G$ if and only if $S\subseteq U$ and $X=U\cap V(G')$ is a vertex cover of $G'$. This implies that $U$ is a minimal connected vertex cover of $G$ if and only if $S\subseteq U$ and $X=U\cap V(G')$ is a minimal vertex cover of $G'$. 

Since $G'$ has at most $3^{n/3}$ minimal vertex covers by Theorem~\ref{thm:vc}, $G$ has at most $3^{n/3}$ minimal connected vertex covers. 
Because $S$ can be found and $G'$ can be constructed in polynomial time, the minimal connected vertex covers of $G$ can be enumerated in time $O^*(3^{n/3})$ using the algorithms of e.g., \cite{TsukiyamaIAS77,JohnsonP88} as mentioned in the preliminaries.
\qed
\end{proof}

Proposition~\ref{prop:lower} shows that the upper bound is tight. 
Now we consider graphs of chordality at most 5. First a definition: a vertex in a graph is \emph{weakly simplicial} if its neighborhood is an independent set and the neighborhoods of its neighbors form a chain under inclusion.

\begin{lemma}[\cite{Pelsmajer04newproofs}]\label{lem:weakly-simp}
A graph is chordal bipartite if and only if every induced subgraph of it has a weakly
simplicial  vertex.  Furthermore,  a  nontrivial  chordal  bipartite  graph  has  a  weakly  simplicial
vertex in each partite set.
\end{lemma}

\begin{observation}\label{obs:contr}
If $e$ is an edge of a graph $G$, then $\chord(G/e)\leq \chord(G)$.
\end{observation}

\begin{theorem}\label{thm:chord-5}
The maximum number of minimal connected vertex covers of a graph of chordality at most 5 is at most $1.6181^n$, and these can be enumerated in time $O(1.6181^n)$.
\end{theorem}

\begin{proof}
We present a  branching algorithm  that we call $\textsc{EnumCVC-chord}(H,S,F)$, which takes as input an induced minor $H$ of $G$ and two disjoint sets $S,F\subseteq V(G)$ and outputs minimal connected vertex covers $U$ of $G$ such that $S\subseteq U\subseteq S\cup F$. We call $\textsc{EnumCVC-chord}(G,\emptyset,V(G))$ to enumerate minimal connected vertex covers of $G$.
As before, we say that $v\in V(G)$ is \emph{free} if $v\in F$ and $v$ is \emph{selected} if $v\in S$. The algorithm branches on a set of free vertices and either selects some of them to be included in a (potential) minimal connected vertex cover or discards some of them  by forbidding them to be selected. 

\medskip
\noindent
$\textsc{EnumCVC-chord}(H,S,F)$
\begin{enumerate}
\item If $S$ is a minimal connected vertex cover then return $S$ and stop. 
If $S$ is a connected vertex cover but not minimal then stop. 
\item If at least two distinct components of $G[S\cup F]$ contain vertices of $S$, then stop.
\item If there are two adjacent free vertices $u,v\in F$, then branch as follows:
\begin{itemize}
\item select $u$, i.e., set $S'=S\cup\{u\}$, $F'=F\setminus\{u\}$, and call \linebreak $\textsc{EnumCVC-chord}(H,S',F')$,
\item discard $u$ and select its neighbors, i.e., set $S'=S\cup N_G(u)$, $F'=F\setminus N_G[u]$, $H'=H-u$, and call  $\textsc{EnumCVC-chord}(H',S',F')$.
\end{itemize}
\item If $F$ is an independent set, then contract 
consecutively 
every edge $uv\in E(H)$ such that $u,v\notin F$ and denote by $H'$ the obtained graph. Find a weakly simplicial 
vertex $u\in V(H')\setminus F$. For each $v\in N_{H'}(u)$, select $v$ and discard $N_{H'}(u)\setminus \{v\}$, i.e., set $S'=S\cup\{v\}$, $F=F\setminus N_{H'}(u)$, $H''=H'-(N_{H'}(u)\setminus\{v\})$, and call  $\textsc{EnumCVC-chord}(H'',S',F')$.  
\end{enumerate}

To show that the algorithm is correct, consider a minimal connected vertex cover $U$ of $G$. Suppose that $S$ and $F$ are  disjoint subsets of $V(G)$, and $H$ is an induced minor of $G$ such that
\begin{itemize}
\item[i)] $S\subseteq U\subseteq S\cup F$,
\item[ii)] for every $v\in V(G)\setminus (S\cup F)$, $N_G(v)\subseteq S$, and
\item[iii] $H$ is obtained from $G$ by deleting vertices of $V(G)\setminus (S\cup F)$ and by contracting some edges $uv$ such that $u,v\in S$. 
\end{itemize}
If $F=\emptyset$, then $U=S$ and the algorithm outputs $U$ on Step~1.
Assume inductively that  $\textsc{EnumCVC-chord}(H',S',F')$ outputs $U$ for any disjoint $S',F'$ and $H'$ satisfying i)--iii) if  $|F'|<|F|$.

Clearly, if $S$ is a connected vertex cover of $G$, then $U=S$ by minimality, and $U$ is returned by the algorithm on Step~1. 
Since $U\subseteq S\cup F$, and $U$ is a connected set in $G$, all the vertices of $S$ are in the same component of $G[S\cup F]$ and the algorithm does not stop at Step~2.

To argue the correctness of Step~3, suppose that there are two adjacent free vertices $u,v\in F$. Then $u\in U$, or $u\notin U$ and $N_G(u)\subseteq U$, because at least one endpoint of every edge is in $U$. In the first case we have that $S'\subseteq U\subseteq S'\cup F'$, where  $S'=S\cup\{u\}$ and $F'=F\setminus\{u\}$, and then we call
 $\textsc{EnumCVC-chord}(H,S',F')$. By induction, the algorithm outputs $U$.  
In the second case,  $S'\subseteq U\subseteq S'\cup F'$, where  $S'=S\cup N_G(u)$ and  $F'=F\setminus N_G[u]$. Also $H'=H-u$; notice that $H'$ is obtained from $H$ by the deletion of a vertex of $V(G)\setminus (S'\cup F')$ and that all neighbors of $u$ are in $S'$. 
Then we call $\textsc{EnumCVC-chord}(H',S',F')$. Again, by induction, the algorithm outputs $U$.  

To consider Step~4, suppose that $F$ is an independent set and $S$ is not a connected vertex cover of $G$. 
Observe that because of ii), $S$ is a vertex cover of $G$, and the vertices of $U\setminus S\subseteq F$ are used to ensure connectivity.
Recall that the graph $H'$ is obtained by contracting edges $uv\in E(H)$ such that 
$u,v\notin F$.  This means that $V(H')\setminus F$ is an independent set, and we have that each vertex of $X=V(H')\setminus F$ is obtained by contracting a component of $G[S]$. 
For each $x\in X$, denote by $W_x\subseteq S$ the set of vertices of the component of $G[S]$ that is contracted to $x$. 
Because the algorithm did not stop at Step~1, $S$ is not a connected vertex cover of $G$ and  $U\setminus S\neq\emptyset$. 
Hence, $F\neq \emptyset$. 
We have that $H'$ is a bipartite graph such that $X,F$ is the bipartition of $V(H')$. 
By Observation~\ref{obs:contr}, $\chord(H')\leq \chord(G)\leq 5$. As $H'$ is bipartite, $\chord(H')\leq 4$, i.e., $H'$ is a chordal bipartite graph.
Notice that because $G[S]$ is disconnected, $|X|\geq 2$.  Because the algorithm did not stop at Step~2, all the vertices of $S$ are in the same component of $G[S\cup F]$. Therefore, $d_{H'}(x)\geq 1$ for $x\in X$. By Lemma~\ref{lem:weakly-simp}, there is a weakly simplicial vertex $u\in X$, and we have that $N_{H'}(u)\neq\emptyset$.

We show that $|N_{H'}(u)\cap U|=1$. Because $U$ is a connected set of $G$ and $G[S]$ is disconnected, $F$ has a vertex that is adjacent to a vertex of the component $G[W_u]$ of $G[S]$. Hence, $N_{H'}(u)\cap U\neq\emptyset$. Since $u$ is a weakly simplicial vertex of $H'$, the neighborhoods of the vertices of $N_{H'}(u)$ form a chain under inclusion. Let $v\in N_{H'}(u)\cap U$ be a vertex with the inclusion maximal neighborhood. Suppose that $(N_{H'}(u)\cap U)\setminus\{v\}\neq\emptyset$ and $w\in (N_{H'}(u)\cap U)\setminus\{v\}$. 
By the choice of $v$, $N_{H'}(w)\subseteq N_{H'}(v)$. Hence, if $w$ is adjacent in $G$ to a vertex of $W_x$ for some $x\in F$, then $v$ is also adjacent to a vertex of $W_x$. 
Because $U$ is a connected set of $G$, we obtain that $U'=U\setminus\{w\}$ is also a connected set. Since $S\subseteq U'$, $U'$ is a vertex cover of $G$, i.e., $U'$ is a connected vertex cover of $G$, but this contradicts the minimality of $U$. Therefore, $N_{H'}(u)\cap U=\{v\}$.

Let $v$ be the unique vertex of $N_{H'}(u)\cap U$. On Step~4 we branch on $v$. We set $S'=S\cup\{v\}$, $F=F\setminus N_{H'}(u)$, and $H''=H'-(N_{H'}(u)\setminus\{v\})$, and we call  $\textsc{EnumCVC-chord}(H'',S',F')$. It remains to observe that the algorithm outputs $U$ for this call by induction.

To complete the correctness proof, it remains to notice that if $S=\emptyset$ and $F=V(G)$, then $S\subseteq U\subseteq S\cup F$ and also $H=G$ is an induced minor of $G$. 
Clearly, i)--iii) are fulfilled for these $S$, $F$ and $H$.
Hence, $\textsc{EnumCVC-chord}(G,\emptyset,V(G))$ outputs $U$. As $U$ is an  arbitrary minimal connected vertex cover, the algorithm outputs all minimal connected vertex covers. 

To obtain the upper bound on the number of minimal connected vertex covers of a graph $G$
of chordality at most 5, it is sufficient to upper bound the number of leaves of the search tree produced by the algorithm.
Recall that $\textsc{EnumCVC-chord}(H,S,F)$ takes as the input a triple $(H,S,F)$. We use $k=|F|$ as the measure of an instance. 
Let $L(k)$ be the maximum number of leaves in the search tree of the algorithm on an instance of size $k$
The algorithm $\textsc{EnumCVC-chord}(H,S,F)$ branches on Steps~3 and 4. 
In Step~3 the algorithm is called recursively for $|F'|=k-1$ on the first branch and for $|F'|\leq k-2$ on the second one. Due to the decrease of the number of free vertices, the branching vector
is $(1,2)$, whose branching number is $\alpha \leq 1.6181$. To analyze the branching in Step~4 let $t=d_{H'}(u)$. Then the algorithm has $t$ branches and in each the new instance has $|F'|=k-t$ free vertices. Hence the branching vector is $(t,t,\ldots ,t)$ with $t\ge 1$ entries and the branching number is $\beta_t=t^{1/t}$.
Notice that if $F=\emptyset$, the algorithm reaches its leaf and we have that $L(0)=1$.
Because on Steps~3 and 4 we never obtain a subproblem with negative measure, we can assume that $L(k)=0$ if $k<0$ and 
$$L(k)\leq\max\{1,L(k-1)+L(k-2),\max\{t\cdot L(k-t)\mid t\geq 1\}\}$$
for $k\geq 1$.
Then  $L(k)\leq \alpha^k$ by Lemma~\ref{lem:recurrences}, because
$\alpha\geq \beta_t$ for $t\geq 1$. 
As $k\leq n$, we have that the number of leaves is at most $\alpha^n$ by Lemma~\ref{lem:upper-bound} and, therefore, $G$ has at most $\alpha^n\leq 1.6181^n$ minimal connected vertex covers.

Because  each step of $\textsc{EnumCVC-chord}$ can be done in polynomial time, the bound on the number of leaves of the search tree  immediately implies that the algorithm runs in time $O(1.6181^n)$.
\qed
\end{proof}

\section{Distance-hereditary graphs}\label{sec:d-h}

Another graph class for which we are able to give a tight upper bound on the maximum number of minimal connected vertex covers, is the class of distance-hereditary graphs. 
First, we need some additional notations.
Let $G$ be a connected graph and $u\in V(G)$. We denote the levels of the breadth-first search (BFS) of $G$ starting at $u$ by $L_0(u),\ldots,L_{s(u)}(u)$. Hence for all 
$i\in \{0,\ldots ,s(u)\}$,
$L_i(u)=\{v\in V(G)\mid\dist_G(u,v)=i\}$. 
Clearly, the number of levels in this decomposition is $s(u)+1$. For $i\in\{1,\ldots,s(u)\}$, we denote by $\mathcal{G}_i(u)$ the set of components of $G[L_i(u)\cup\ldots \cup L_{s(u)}(u)]$, and $\mathcal{G}(u)=\cup_{i=1}^{s(u)}\mathcal{G}_i(u)$.
Let $H\in \mathcal{G}_i(u)$ and $B=N_G(V(H))$. Clearly, $B\subseteq L_{i-1}(u)$. We say that $B$ is the \emph{boundary of $H$ (in $L_{i-1}(u)$)}. We also say that $I=L_i(u)\cap V(H)$ is the \emph{interface} of $H$ (in $L_i(u)$).
For $i\in\{0,\ldots,s(u)-1\}$, $\mathcal{B}_i(u)$ is the set of boundaries in $L_i(u)$ of the graphs of $\mathcal{G}_{i+1}(u)$ and $\mathcal{B}(u)=\cup_{i=0}^{s(u)-1}\mathcal{B}_i(u)$.
We will use the following result due to Bandelt and Mulder~\cite{BandeltM86}, and D'Atri and Moscarini~\cite{DAtriM88}. 

\begin{lemma}[\cite{BandeltM86,DAtriM88}]\label{lem:d-h}
A connected graph $G$ is distance-hereditary if and only if for every vertex $u\in V(G)$ and every $H\in\mathcal{G}(u)$ with boundary $B$, the following holds: $N_G(u)\cap V(H)=N_G(v)\cap V(H)$ for $u,v\in B$.
\end{lemma}

For the main result of this section, we need the following structural properties of distance-hereditary graphs.

\begin{lemma}\label{lem:nested}
Let $G$ be a connected distance-hereditary graph and $u\in V(G)$. Then for any $B_1,B_2\in\mathcal{B}(u)$, either $B_1\cap B_2=\emptyset$ or $B_1\subseteq B_2$ or $B_2\subseteq B_1$.
\end{lemma}

\begin{proof}
Assume for contradiction that there are $B_1,B_2\in\mathcal{B}(u)$ such that $B_1\cap B_2\neq\emptyset$ but neither $B_1\subseteq B_2$ nor $B_2\subseteq B_1$. 
Let $B_1$ and $B_2$ be the boundaries of $H_1$ and $H_2$ respectively, and let $I_1$ and $I_2$ be the interfaces of $H_1$ and $H_2$ respectively.
Let $x\in I_1$ and $y\in I_2$.
By Lemma~\ref{lem:d-h}, $H_1\neq H_2$ and $\dist_G(x,y)=2$. Consider $G'=G-B_1\cap B_2$. Because $B_1\setminus B_2\neq\emptyset$ and $B_2\setminus B_1=\emptyset$, $x$ and $y$ are in the same component of $G'$, because $G'$ has paths that connect $u$ to  $x$ and $y$ respectively, but $\dist_{G'}(x,y)\geq 3$, which gives the desired contradiction.\qed
\end{proof}

\begin{lemma}\label{lem:twins}
Let $G$ be a connected distance-hereditary graph, $u\in V(G)$ and let $B$ be an inclusion minimal set of $\mathcal{B}(u)$. If $B$ is an independent set of $G$, then the vertices of $B$ are false twins.
\end{lemma}

\begin{proof}
If $|B|=1$, then the claim is trivial. Assume that $|B|\geq 2$. Then $B\subseteq L_i(u)$ for some $i\in \{1,\ldots,s(u)-1\}$. 
Let $B$ be the boundary of $H\in \mathcal{G}_{i+1}(u)$ with the interface $I$.
Consider two distinct $x,y\in B$. 
To obtain a contradiction, assume that $N_G(x)\neq N_G(y)$ and $z\in N_G(x)\setminus N_G(y)$.
By Lemma~\ref{lem:d-h}, $N_G(x)\cap L_{i-1}(u)=N_G(y)\cap L_{i-1}(u)$. Because $B$ is inclusion minimal, by Lemmas~\ref{lem:d-h} and \ref{lem:nested}, $N_G(x)\cap L_{i+1}(u)=N_G(y)\cap L_{i+1}(u)$. Hence, $z\in L_i(u)$. By Lemma~\ref{lem:d-h}, there is $v\in L_{i-1}(u)$ such that $vx,vy,vz\in E(G)$ and there is $w\in I$ such that $xw,yw\in E(G)$. Because $B$ is an independent set, $z\notin B$ and $zw\notin E(G)$. It remains to observe that the set of vertices $\{v,x,y,x,w\}$ induces a subgraph that is known as the \emph{house}, which is a forbidden induced subgraph of distance-hereditary graphs~\cite{BandeltM86,DAtriM88}. Thus $G$ cannot be distance-hereditary, and the obtained contradiction proves the lemma.\qed
\end{proof}

\begin{observation}\label{obs:Knm}
Let $G$ be a graph, and let $X,Y\subseteq V(G)$ be disjoint sets such that every vertex of $X$ is adjacent to every vertex of $Y$. Then for every vertex cover $U$ of $G$, $X\subseteq U$ or $Y\subseteq U$. 
\end{observation}

\begin{proof}
Suppose that for $U\subseteq V(G)$, $x\in X\setminus U\neq\emptyset$ and $y\in Y\setminus U\neq\emptyset$. Then $U$ does not cover $xy$, i.e., is not a vertex cover of $G$.\qed 
\end{proof}

\begin{theorem}\label{thm:d-h}
The maximum number of minimal connected vertex covers of a distance-hereditary graph is at most $2 \cdot 3^{n/3}$, and these can be enumerated in time $O^*(3^{n/3})$.
\end{theorem}

\begin{proof}
Let $G$ be a distance-hereditary graph. If $G$ has no edges, then the claim is trivial. Notice also that the removal of an isolated vertex does not influence connected vertex covers, and if $G$ has two components with at least one edge each, $G$ has no connected vertex cover. Hence, without loss of generality we can assume that $G$ is a connected graph and $n\geq 2$.  

Let $u\in V(G)$. For the main part of the proof, we give an algorithm for enumerating all minimal connected vertex covers of $G$ that contain $u$ and upper bound the number of such covers. At the end, we argue why this is sufficient. 

First we perform breadth-first search of $G$ starting at $u$ and construct $\mathcal{G}(u)$ and  $\mathcal{B}(u)$. 
We construct the set 
$\mathcal{G}'(u)\subseteq \cup_{i=2}^{s(u)}\mathcal{G}_i(u)$ 
that contains all 
$H\in \mathcal{G}(u)$ 
such that the boundary $B$ of $H$ is an inclusion minimal set of $\mathcal{B}(u)$. 

Next we give a branching algorithm that we call
$\textsc{EnumCVC-d-h}(R,S,F)$, 
which takes as input an induced subgraph $R$ of $G$ and two disjoint sets $S,F\subseteq V(G)$ such that $u\in S$,  and outputs minimal connected vertex covers $U$ of $G$ such that $S\subseteq U\subseteq S\cup F$. To enumerate all minimal connected vertex covers $U$ of $G$ such that $u\in U$, we call $\textsc{EnumCVC-d-h}(G,\{u\},V(G)\setminus\{u\})$. As before, we say that $v\in V(G)$ is free if $v\in F$ and $v$ is selected if $v\in S$. The algorithm branches on a subset of the free vertices and either selects some of them to be included in a (potential) minimal connected vertex cover or discards some of them  by forbidding them to be selected. 

\medskip
\noindent
$\textsc{EnumCVC-d-h}(R,S,F)$

\medskip
\noindent
{\bf Step~1.}  
If $S$ is a minimal connected vertex cover then return $S$ and stop.

\medskip
\noindent
{\bf Step~2.} If there is an $H\in\mathcal{G}'(u)$ with  boundary $B$ such that $B\cap S=\emptyset$ and $B$ is an independent set, then let $B'=B\cap V(R)$ and do the following.
 \begin{itemize}[leftmargin=8mm]
\item[2.1.] If $B'=\emptyset$, then stop.
\item[2.2.] If $|B'|=1$, then select the unique vertex of $B'$, i.e., set $S'=S\cup B'$, $F'=F\setminus B'$, and call $\textsc{EnumCVC-d-h}(R,S',F')$.
\item[2.3.] If $|B'|\geq 2$ and $N_R(B')\cap F=\emptyset$, then for each $v\in B'$ branch by selecting $v$ and discarding $B'\setminus\{v\}$, i.e., set 
$S'=S\cup \{v\}$, $F'=F\setminus B'$, $R'=R-(B'\setminus\{v\})$, and call $\textsc{EnumCVC-d-h}(R',S',F')$.
\item[2.4.] If $|B'|\geq 2$ and $|N_R(B')\cap F|=1$, then denote by $w$ the unique vertex of  $N_R(B')\cap F$ and branch:
\begin{itemize}
\item for each $v\in B'$, select the vertices $v,w$ and discard $B'\setminus\{v\}$, i.e., set 
$S'=S\cup \{v,w\}$, $F'=F\setminus (B'\cup\{w\})$, and 
$R'=R-(B'\setminus\{v\})$, then call $\textsc{EnumCVC-d-h}(R',S',F')$,
\item select the vertices of $N_R(w)\supseteq B\rq{}$ and discard $w$, i.e., set $S'=S\cup N_R(w)$, $F'=F\setminus N_R[w]$, $R'=R-\{w\}$, and call  $\textsc{EnumCVC-d-h}(R',S',F')$.
\end{itemize}
\item[2.5.] If   $|B'|\geq 2$ and   $|N_R(B')\cap F|\geq 2$, then branch:
\begin{itemize}
\item for each $v\in B'$, select the vertices of $\{v\}\cup (N_R(B')\cap F)$ and discard the 
vertices of $B'\setminus\{v\}$, i.e., set 
$S'=S\cup \{v\}\cup N_R(B')$, $F'=F\setminus N_R[B']$, $R'=R-(B'\setminus\{v\})$, and call $\textsc{EnumCVC-d-h}(R',S',F')$,
\item select all vertices of $B'$, i.e., set $S'=S\cup B'$, $F'=F\setminus B'$, and call  $\textsc{EnumCVC-d-h}(R,S',F')$.
\end{itemize}
\end{itemize}

\medskip
\noindent
{\bf Step~3.}  Let $R'=G[F]$. Then enumerate the minimal vertex covers of $R'$, and for each vertex cover $X$ of $R'$, output $S\cup X$ if $S\cup X$ is a minimal connected vertex cover of $G$.

\medskip

To prove the correctness, let $U$ be a minimal connected vertex cover of $G$ such that $u\in U$, and consider two disjoint subsets $S$ and $F$ of $V(G)$ and $R=G[S\cup F]$  such that
\begin{itemize}
\item[i)] $u\in S\subseteq U\subseteq S\cup F$,
\item[ii)] for any $v\in V(G)\setminus V(R)$, $N_G(v)\subseteq S$.
\end{itemize}
Notice that for $S=\{u\}$ and $F=V(G)\setminus\{u\}$, i) and ii) are fulfilled.  
Assume inductively that  $\textsc{EnumCVC-d-h}(R',S',F')$ outputs $U$ for any disjoint $S',F'\subseteq V(G)$ and $R'=G[S'\cup F']$ such that
i) and ii) are fulfilled for $S',F'$ and $|F'|<|F|$.

Clearly, if $S$ is a connected vertex cover of $G$, then $U=S$ by minimality and $U$ is returned by the algorithm on Step~1. 

Now we consider Step~2. Suppose that there is $H\in\mathcal{G}'(u)$ with the boundary $B$ such that $B\cap S=\emptyset$ and $B$ is an independent set. By Lemma~\ref{lem:twins}, the vertices of $B$ are false twins, i.e., each vertex of $B$ is adjacent to every vertex of $N_G(B)$ in $G$. By Observation~\ref{obs:Knm}, $B\subseteq U$ or $N_G(B)\subseteq U$.
Notice that if there is a $v\in N_G(B)\setminus N_R(B')$, then $B\subseteq N_G(v)\subseteq S$ by ii). Since, $B'\cap S=\emptyset$, we have that $N_G(B)=N_R(B')$. 

Since $B$ is a separator of $G$, $U\cap B\neq\emptyset$ by Observation~\ref{obs:cuts}. In particular, it means that the algorithm does not stop at Step~2.1 and $B'=B\cap V(R)\neq\emptyset$. 

If $|B'|=1$, then the unique vertex of $B'$ is in $U$, because  $U\cap B\neq\emptyset$. Clearly, i) and ii) are fulfilled for $S'=S\cup B'$, $F'=F\setminus B'$, and  $\textsc{EnumCVC-d-h}(R,S',F')$ outputs $U$ by induction on Step~2.2.

From now we may assume that $|B'|\geq 2$.

Suppose that $N_R(B')\cap F=\emptyset$.  
Recall that $N_R(B')=N_G(B)$. Hence, $N_G(B)\subseteq S$
and all the vertices incident to the vertices of $B$ are covered by $S$. Hence, a vertex of $B$ can be included in $U$ only to ensure connectivity. It implies that exactly one vertex $v\in B'$ is in $U$.  We have that for  $S'=S\cup \{v\}$ and $F'=F\setminus B'$ i) and ii) are fulfilled,  $R'=R-(B'\setminus\{v\})=G[S'\cup F']$, and $\textsc{EnumCVC-d-h}(R',S',F')$ outputs $U$ by induction on Step~2.3.

Suppose that $N_R(B')\cap F=\{w\}$. Because  $N_R(B')=N_G(B)$, $w$ is the unique vertex of $N_G(B)\cap F$ and $N_G(B)\setminus\{w\}\subseteq S$.
Recall that $B\subseteq U$ or $N_G(B)\subseteq U$.
Suppose that $N_G(B)\subseteq U$. Then $w\in U$. 
We have that at least one vertex of $B$ is in $U$. Notice that all the edges incident to the vertices of $B$ are covered by the vertices of $S$ and $w$. Hence, exactly one vertex $b\in B'$ is in $U$ to ensure connectivity by minimality.  We have that i) and ii) are fulfilled for 
$S'=S\cup \{v,w\}$ and $F'=F\setminus (B'\cup\{w\})$,   $R'=R-(B'\setminus\{v\})=G[S'\cup F']$, and $\textsc{EnumCVC-d-h}(R',S',F')$ outputs $U$ by induction.
Suppose that $B\subseteq U$ but $N_G(B)\not\subseteq U$. Since $N_G(B)\setminus\{w\}\subseteq S$, $w\notin U$. Observe that $B'\subseteq N_R(w)$.
We obtain  that  i) and ii) are fulfilled for  $S'=S\cup N_R(w)$, $F'=F\setminus N_R[w]$, $R'=R-w=G[S'\cup F']$, and  $\textsc{EnumCVC-d-h}(R',S',F')$ outputs $U$ by induction. We conclude that if $N_R(B')\cap F=\{w\}$, then the algorithm outputs $U$ on Step~2.4.

Finally let us consider the last case of Step~2. Assume that $|N_R(B')\cap F|\geq 2$. Recall that $N_G(B)=N_R(B')$ and $B\subseteq U$ or $N_G(B)\subseteq U$.
If $N_G(B)\subseteq U$, then by the same arguments as above, 
exactly one $v\in B'$ is in $U$. We have that i) and ii) are fulfilled for 
$S'=S\cup \{v\}\cup N_R(B')$ and $F'=F\setminus N_R[B']$, $R'=R-(B'\setminus\{v\})=G[S'\cup F']$, and $\textsc{EnumCVC-d-h}(R',S',F')$ outputs $U$ by induction. Assume that $B\subseteq U$. Then i) and ii) are fulfilled for $S'=S\cup B'$, $F'=F\setminus B'$, and  $\textsc{EnumCVC-d-h}(R,S',F')$ outputs $U$ by induction. Hence, the algorithm outputs $U$ at Step~2.5.

It remains to consider Step~3. Because $U$ is a vertex cover of $G$ and $S\subseteq U\subseteq S\cup F$, there is a minimal vertex cover $X$ of $G[F]$ such that $X\subseteq U$. We claim that $U'=S\cup X$ is a connected vertex cover. By ii), we have that $U'$ is a vertex cover of $G$. 

To show connectivity, we prove that for every $v\in U'$, there is a $(u,v)$-path in $G[U']$. Clearly, $v\in L_i$ for some $i\in\{0,\ldots,s(u)\}$. We prove the claim by induction on $i$. If $i=0$, then 
$u\in L_0$, and if $i=1$, then $v\in N_G(u)$. Suppose that $v\in L_i$ for $i\geq 2$. Then $v$ is in the interface $I$ of some $H\in \mathcal{G}_i(u)$. Let $B$ be the boundary of $H$.  Then there is an $H'\in \mathcal{G}'(u)$ with the boundary $B'\subseteq B$. If there is a $w\in B'\cap U$, then $vw\in E(G)$ and $G[U']$ has a $(u,w)$-path by induction.
Clearly, this path can be extended to $v$. Assume that $B'\cap U=\emptyset$. If $B'$ has two adjacent vertices then at least one of them is in the vertex cover $U'$. Hence $B'$ is an independent set.  Since  $B'\cap U=\emptyset$, $B'\cap S=\emptyset$. This contradict the fact that the algorithm goes to Step~3, as it should execute Step~2 if $B'\cap S=\emptyset$ and $B'$ is independent. 

We obtain that $U'$ is a connected vertex cover. By minimality, $U=U'$, i.e., the algorithm outputs $U$ on Step~3.

We proved that $\textsc{EnumCVC-d-h}(G,\{u\},V(G)\setminus\{u\})$ enumerates all minimal connected vertex covers of $G$ that include $u$.
To obtain the upper bound on the number of such minimal connected vertex covers of $G$, it is sufficient to upper bound the number of leaves of the search tree produced by the algorithm. We assume that $k=|F|$ is the measure of an instance $(H,S,F)$.
The algorithm $\textsc{EnumCVC-d-h}(H,S,F)$ branches on Steps~2.3--5 and we can assume that the enumeration of minimal connected vertex covers on Step~3 is done by the straightforward modification of branching algorithm \textsc{mis1} from \cite[Section~1.3]{FominK10}\footnote{The algorithm  \textsc{mis1} is used in \cite[Section~1.3]{FominK10} to find a maximum independent  set, but it can be easily modified to enumerate all maximal independent sets.}. 
Note that the 
branching vectors used in  algorithm  \textsc{mis1} are  $(t,t,\ldots t)$ with
$t\geq 1$ entries; denote the corresponding branching number $\alpha_t=t^{1/t}$. 
On Step~2.3 we have $h=|B'|$ branches, and for each branch $|F'|\leq |F|-h$, thus a branching vector of $(h,h,\ldots ,h)$ with $h\ge 2$ entries and the branching number is $\alpha_h$.
On Step~2.4 we have $s+1$ branches for $s=|B'|$, and for each branch $|F'|\leq |F|-(s+1)$,
hence a branching vector of $(s+1,s+1,\ldots ,s+1)$ with $s+1\ge 3$ entries and the corresponding branching numbers is $\alpha_{s+1}$. 
On Step~2.5, we have $t=|B'|\geq 2$ branches with $|F'|\leq |F|-t-2$ and one branch with $|F'|\leq |F|-t$. Hence we get the branching vector $(t+2,t+2,\ldots ,t+2,t)$ with $t+1\ge 3$ 
entries and the branching number $\beta_{t}$. 
Let $L(k)$ be the maximum number of leaves in the search tree of the algorithm on a graph $G$ of $k$ free vertices. 
If $F=\emptyset$, the algorithm reaches its leaf and we have that $L(0)=1$.
Because we never obtain a subproblem with negative measure, we can assume that $L(k)=0$ if $k<0$ and 
$$L(k)\leq\max\{1,\max\{t\cdot L(k-t)\mid t\geq 1\},\max\{t\cdot L(k-t-2)+L(k-t)\mid t\geq 2\}\}$$
for $k\geq 1$.
Because $\alpha_3\geq \alpha_t$ for $t\geq 1$ and $\alpha_3\geq \beta_t$ for $t\geq 2$,
$L(k)\leq\alpha^k$ by Lemma~\ref{lem:recurrences}.
Since $k\leq n$, we obtain that the number of leaves of the search tree is at most $3^{n/3}$ by Lemma~\ref{lem:upper-bound}.

This gives an upper bound on the number minimal connected vertex covers of $G$ that contain $u$. It remains to observe that if $u_1u_2\in E(G)$, then every vertex cover of $G$ contains $u_1$ or $u_2$. Hence, the number of minimal connected vertex covers of $G$ is upper bounded by the sum of  the numbers of minimal connected vertex covers that contain $u_1$ and $u_2$ respectively. 

Because the  breadth-first search of $G$ starting at $u$ and constructing $\mathcal{G}'(u)$ can be done in polynomial time, and 
 each step of $\textsc{EnumCVC-d-h}$ also can be done in polynomial time, the bound on the number of leaves of the search tree  immediately implies that the algorithm runs in time $O^*(3^{n/3})$.
\qed
\end{proof}

\medskip

Proposition~\ref{prop:lower} shows that the upper bound is tight up to a constant factor.

\section{Conclusions}

The bounds that we have given
for chordal graphs and distance-hereditary graphs are tight. While we can hope to improve the other upper bounds of this paper, we conjecture that they exceed $3^{n/3}$. It can be observed that for some more narrow classes of graphs of bounded chordality, the number of minimal connected vertex covers becomes polynomial.

\begin{proposition}\label{prop:split} 
The number of minimal connected vertex covers of a split graph  $G$ is at most $n$, and these can be enumerated in time $O(n+m)$. The number of minimal connected vertex covers of a cobipartite graph $G$ is at most $n^2/4+n$, and these can be enumerated in time $O(n^2)$.
\end{proposition}

\begin{proof}
Notice that if $X$ is a clique of a graph $G$, then for every vertex cover $U$ of $G$, either $X\subseteq U$ or $|X\setminus U|=1$.
Let $G$ be a split graph. Without loss of generality we can assume that $G$ is a connected graph with at least two vertices. Let $K,I$ be a partition of $V(G)$ in a clique $K$ and an independent set $I$ and assume that $K$ is an inclusion maximal clique of $G$. 
 If $V(G)=K$, then $G$ has $n$ minimal connected vertex covers $K\setminus \{v\}$ for $v\in V(G)$. Assume  that $I\neq\emptyset$. Then $K$ is a connected vertex cover of $G$. For $v\in K$, if $U$ is a minimal connected vertex cover of $G$ with $v\notin U$, $U=(K\setminus\{v\})\cup N_G(v)$. It immediately implies that $G$ has at most $n$ connected vertex covers. Taking into account that a partition $K,I$ can be found in time $O(n+m)$, it follows that the minimal connected vertex covers can be enumerated in time $O(n+m)$.

Let now $G$ be a  cobipartite graph. Again, we can assume without loss of generality that $G$ is a connected graph with at least two vertices. If $G$ is a complete graph, then $G$ has $n$ connected vertex covers. Assume that $G$ is not a complete graph, and let $K_1,K_2$ be a partition of $V(G)$ into two cliques. Let $U$ be a minimal connected vertex cover of $G$.
If $K_1\subseteq U$, then $U=V(G)\setminus \{v\}$ for $v\in K_2$, and there are at most $|K_2|$ sets of this type. Symmetrically, there are at most $|K_1|$ minimal connected vertex covers $U$ with $K_2\subseteq U$. If $K_1\setminus U\neq \emptyset$ and $K_2\setminus U\neq\emptyset$, then $U=V(G)\setminus\{u,v\}$ for $u\in K_1$ and $v\in K_2$, and $G$ has at most $|K_1||K_2|$ such  minimal connected vertex covers. We conclude that $G$ has at most $|K_1|+|K_2|+|K_1||K_2|\leq n^2/4+n$ minimal connected vertex covers. Clearly, these arguments can be applied to obtain an enumeration algorithm that runs in time $O(n^2)$.\qed
\end{proof}

Finally let us consider a related combinatorial question.  What is the 
relation of the maximum number of minimal vertex covers and the maximum number of minimal connected vertex covers?  More precisely, is there a particular class of connected graphs  for which the difference of both numbers is exponential in $n$? 
On one hand, trees provide an easy answer. Clearly every tree has a unique minimal connected vertex cover consisting of all its cut vertices.  On the other hand, there are trees with at least 
$2^{\frac{n-1}{2}}$ minimal vertex covers; e.g. a union of $K_2$'s with an additional vertex adjacent to exactly one vertex of each $K_2$.  Hence there are trees on $n$ vertices having one minimal connected vertex cover but their number of minimal vertex covers is exponential in $n$. 

It is more interesting to observe that the number of minimal connected vertex covers of 
a connected graph may be  significantly larger than its number of minimal vertex covers. Consider the graphs $G_k$, $k\ge 1$ integer, constructed as follows.
\begin{itemize}
\item The vertex set of $G_k$ consists of $A_i=\{a_i,b_i,c_i,d_i,e_i\}$ for all $i=1,2,\ldots k$ and 
$x_i$ for all  $i=1,2,\ldots ,k+1$. Thus $G_k$ has $n=6k+1$ and $G_1$ has $7$ vertices.   
\item For all $i=1,2,\ldots ,k$, vertex $x_i$ is adjacent in $G_k$ to all vertices of $A_i\cup A_{i+1}$ 
and $x_{k+1}$ is adjacent to all vertices of $A_k$.
\end{itemize}
On one hand, observe that for every $A_i$, every minimal vertex cover of $G_k$ contains either all vertices of $A_i$ or no vertex of $A_i$. Hence the minimal vertex covers of $G_k$
are in one-to-one correspondence to the minimal vertex covers of the path $P_{2k+1}$
obtained by identifying in $G_k$ all vertices of $A_i$ for all $i=1,2,\ldots ,k+1$. Since a path is triangle-free, both  $P_{2k+1}$ and $G_k$  have at most $2^{\frac{2k+1}{2}} = 2^{\frac{n+2}{6}}$ minimal vertex covers~\cite{HujteraT93}.
On the other hand, all cut vertices of $G_k$, i.e. $x_2,x_3,\ldots, x_k$ belong to every
minimal connected vertex cover $X$ of $G_k$. Thus for every minimal connected vertex cover 
$X$ we have $|X\cap A_i|=1$ for all $i=1,2,\ldots k$, which implies that $G_k$ has at least 
$5^k= 5^{\frac{n-1}{6}}$. Since $2^{1/6} < 5^{1/6}$ and $n\ge 7$, the difference 
between the number of minimal connected vertex covers and the number of minimal  vertex covers is exponential in $n$. 

It is an open question whether there is a family of connected graphs such that their number 
of minimal vertex covers is polynomial in $n$, while their number of minimal connected
vertex covers is exponential in $n$.

\end{document}